\documentclass[11pt]{article}
\usepackage{graphicx}
\usepackage{amsmath}
\RequirePackage{amssymb}
\RequirePackage{amsthm}
\usepackage{amsmath,scalefnt}
\usepackage{footnote}
\usepackage{fmtcount}
\usepackage{color}
\usepackage{threeparttable}
\usepackage{chemarr}

	\theoremstyle{plain}

	\newtheorem{lemma}{Lemma}

	 \newtheorem{proof of proposition}{Proof of Proposition}

	\theoremstyle{definition}

\renewcommand{\r}{{\mathbb R}}

\newcommand{\abs}[1]{|#1|}

\newcommand{\beqn}{\begin{eqnarray*}}
\newcommand{\eeqn}{\end{eqnarray*}}
\newcommand{\pic}[2]{\includegraphics[scale=#1]{#2}}
\newcommand{\picc}[2]{\begin{center}\pic{#1}{#2}\end{center}}
\newcommand{\be}[1]{\begin{equation}\label{#1}}
\newcommand{\ee}{\end{equation}}
\newcommand{\bi}{\begin{itemize}}
\newcommand{\ei}{\end{itemize}}
\newcommand{\ben}{\begin{enumerate}}
\newcommand{\een}{\end{enumerate}}

\title{Remarks on a population-level model of chemotaxis:\\
advection-diffusion approximation and simulations}
\author{{Zahra Aminzare and Eduardo D. Sontag}\\ \\
{\small Department of Mathematics, Rutgers University,}\\
{\small Piscataway, NJ 08854-8019 USA}}

\begin{document}
\maketitle

\begin{abstract}
This note works out an advection-diffusion approximation to the density of a
population of \emph{E.\ coli} bacteria undergoing chemotaxis in a
one-dimensional space.
Simulations show the high quality of predictions under a shallow-gradient
regime.
\end{abstract}

\section{Introduction} \label{introduction}

The chemotactic behavior of \emph{E.\ coli} has been studied widely at both
the microscopic and macroscopic levels. The movement of bacteria involves a
directed movement (run) and a random turning (tumble). 
Each individual carries an internal state which, in the presence of a time and
space dependent external signal, may be modeled as evolving according to a
system of ordinary differential equations.
In the presence of a signal (typically, a nutrient) in the environment, the
individual changes its direction%
\footnote{Potentially, its speed may change too, though this seems not to be a
  substantial factor in \emph{E.\ coli} chemotaxis.}
at random, with a tumbling rate which depends on the internal state, biasing
moves toward more favorable environments or away from noxious substances.
This random reorientation introduces a stochastic character to the evolution
equations, and the population behavior of such hybrid systems is modeled by
jump-Markov state-dependent systems.

The transport equation that describes a jump-markov system is very difficult to study
mathematically, and cannot be validated by typical experimental techniques
such as optical density measurements of bacteria in microfluidics chambers.
Thus it is of great interest to derive a simpler macroscopic
equation for the density of bacteria from microscopic equations.
In addition, the current interest in scale invariant transient behavior
(``fold-change detection,'' see for example
\cite{shoval10,stocker_shimizu11,shoval_alon_sontag_2011})
requires such macroscopic descriptions when starting from the jump-Markov
model, as discussed in~\cite{11cdc_shoval_alon_sontag}.
Our goal in this note is to work out, in a shallow-gradient regime, the tools
developed by Erban, Othmer, and Grunbaum
\cite{Erban_Othmer_2004,Grunbaum_2000} for a mechanistically realistic model of
\emph{E.\ coli} chemotaxis.
In the case of exponential gradients, the result is a constant-coefficient
advection-diffusion equation.
We provide the calculations as well as an agent-based simulation that verifies
the theoretical predictions, for one-dimensional motions.
Future work will expand these considerations to two and three dimensions.

\subsection{Preliminaries}

Let $p(x, y, \nu, t)$ be a density function describing a population of agents
(for example, bacteria), modeled in a $2N+m$ dimensional phase space, where at
time $t$, $x=(x_1,\ldots,x_N)\in \r^N$ ($N=1,2,3$; we soon specialize to
$N=1$) denotes the position of a cell,
$y=(y_1, \ldots,y_m)\in Y\subset \r^m $ denotes the internal dynamics of a
cell, and $\nu\in V\subset\r^N$ denotes its velocity.
Also, $S(x, t)=(S_1, \ldots, S_M)\in \r^M$ denotes the concentration of
extracellular signals in the environment. 

We assume that the following system of ordinary differential equations
describes the evolution of the intracellular state, in the presence of the
extracellular signal $S$:
\begin{equation}
\label{deterministic_model}
\frac{dy}{dt}\;=\; f(y, S)
\end{equation}
where $f\colon \r^m\times\r^M\to \r^m$ is a continuously differentiable
function with respect to each component, i.e.,  $f\in C^1(\r^m\times\r^M)$. 

The evolution of $p$ with turning rate $\lambda\;=\;\lambda(y)$ is governed by
the following transport (or ``Fokker-Planck'' or ``forward Kolmogorov'')
equation: 
\begin{equation}
\label{transport}
\frac{\partial p}{\partial t}+\nabla_x\cdot \nu p +\nabla_y\cdot fp = -\lambda(y)p +\displaystyle\int_V \lambda(y) T(y, \nu, \nu')p(x, y, \nu', t)\;d \nu'
\end{equation}
where the nonnegative kernel $T(y, \nu, \nu')$ is the probability that the bacteria changes the velocity from $\nu'$ to $\nu$ if a change of direction occurs. Also 
\[\displaystyle\int_V T(y, \nu, \nu') \;d\nu=1.\]

The main goal of this note is to derive a macroscopic model for chemotaxis
using the microscopic model (\ref{transport}), i.e., we want to find an
equation to describe the evolution of the marginal density:
\begin{equation}
n(x, t)\;=\;\displaystyle\int_V\displaystyle\int_Y p(x, y, \nu, t)\;dy d\nu.
\end{equation}
As remarked in the introduction, this is of interest both because of
experimental and theoretical reasons, in particular in the context of
scale-invariant sensing~\cite{11cdc_shoval_alon_sontag}.

\section{Chemotaxis equation in one dimensional movement}

In this section, for simplicity, we study the movement of agents in one
dimension and assuming a one-dimensional state, i.e., $m=N=M=1$.
We also take the speed, $\nu$, as constant.
Let $p^{\pm}(x, y, t)$ denote the density of the particles that at time $t$,
are located at position $x$, with the internal state $y$, and with the
constant speed $\nu$, and moving to the right ($+$) or left ($-$)
respectively. 
We assume that the following decay condition:
\begin{equation}\label{condition_p+-}
p^{\pm}(x, y, t)\;\leq\; A(x, t)e^{-\alpha(x,t) y},
\end{equation}
for some functions $A, \alpha\colon\r\times[0, \infty)\to\r_{>0}$.
The internal state evolves according to the following ODE system:
\begin{equation}
\label{deterministic_model_general}
\frac{dy}{dt}\;=\;f^{\pm}(y, \nu, S, S'),
\end{equation}
where $f^{\pm}\colon\r\times\r\times\r\times\r\to\r$ are continuously
differentiable functions in each argument that describe the evolution of
internal state of bacteria which move to the right ($+$) and left ($-$)
respectively.  
 
 Note that we are allowing $f$ to depend on the direction of movement as well as $\nu$ and $S'$, the derivative of $S$ with respect to space. In our example, $f^+=f^-$ only depends on $y$ and $S$, but we can consider the more general dependence in these preliminary derivations. 
 
We consider the following equation for the
tumbling rate: 
 \begin{equation}
 \label{response_general}
 \lambda(y, S, S')\;=g(y, S, S'),
\end{equation}
for some continuous function $g$.
 
Then, according to Equation (\ref{transport}), $p^{\pm}(x, y, t)$ satisfy the
following coupled first-order partial differential equations: 
\begin{align}
\frac{\partial p^+}{\partial t}\;+\;&\nu\frac{\partial p^+}{\partial x} \;+\; \frac{\partial }{\partial y} \left[f^{+}(y, \nu, S, S')\; p^+\right]\label{transport_one_dim_general:1}
\;=\; g(y, S, S') (-p^++p^-)\\
\frac{\partial p^-}{\partial t}\;-\;&\nu\frac{\partial p^-}{\partial x} \;+\; \frac{\partial }{\partial y} \left[f^{-}(y, \nu, S, S')\; p^-\right]\label{transport_one_dim_general:2}
\;=\; g(y, S, S') (p^+-p^-).
\end{align}
We now state the following lemma from \cite{Erban_Othmer_2004} regarding the
existence and uniqueness of solutions of
(\ref{transport_one_dim_general:1})-(\ref{transport_one_dim_general:2}):

\begin{lemma}[Existence and uniqueness of solutions]
 Suppose that $f^{\pm}\in C^1(\r\times\r\times\r\times\r)$, and let $S\colon\r\times[0,\infty)\to\r$ be continuous. In addition, assume that $g$ in (\ref{response_general}) is always nonnegative, and $p^{\pm}_0\colon\r\to[0, \infty)$ are given nonnegative compactly supported $C^1$ functions. Then there exists a domain $\mathcal{Q}\subset\r\times[0, \infty)$ containing the entire line $t=0$ such that the system of equations (\ref{transport_one_dim_general:1})-(\ref{transport_one_dim_general:2}) with initial conditions $p_0^{\pm}$ has a unique $C^1$ solution in $\mathcal{Q}$. Moreover, the functions $p^{\pm}$ are nonnegative wherever they are defined. 
 \end{lemma}
 
 The objective is to derive an equation for the macroscopic density function 
 \begin{equation}
 \label{density}
 n(x,t)\;=\;\displaystyle\int_{\r} p^+(x, y, t)+p^-(x, y, t)\;dy,
 \end{equation}
 using the microscopic model (\ref{transport_one_dim_general:1})-(\ref{transport_one_dim_general:2}), by the following technique from \cite{Erban_Othmer_2004}.
 To this end we define the following additional moments: 
 \begin{equation}
 \label{moments}
\aligned
 n_i(x,t)\;&=\;\displaystyle\int_{\r}y^i\left( p^+(x, y, t)+p^-(x, y, t)\right)\;dy,\quad\mbox{for}\quad i= 1, 2, \ldots\\
 j(x,t)\;&=\;\displaystyle\int_{\r} \nu\left(p^+(x, y, t)-p^-(x, y, t)\right)\;dy,\\
 j_i(x,t)\;&=\;\displaystyle\int_{\r}y^i \nu\left(p^+(x, y, t)-p^-(x, y, t)\right)\;dy,\quad\mbox{for}\quad i= 1, 2, \ldots.
\endaligned
\end{equation}
Note that by condition (\ref{condition_p+-}) all the moments are well defined.

Next, we assume 
\begin{equation}
\aligned
f^+\;&=\;f_0+\nu f_1,\\
f^-\;&=\;f_0-\nu f_1,
\endaligned
\end{equation}
where the Taylor expansions of $f_0$ and $f_1$ are given as follows:
\begin{align}
f_0\;=\;A_0+A_1 y+A_2 y^2+\ldots,\label{taylor:f0}\\
f_1\;=\;B_0+B_1 y+B_2 y^2+\ldots,\label{taylor:f1}
\end{align}
for some $A_i$'s and $B_i$'s that are functions of $S$, $S'$, and $\nu^2$.

Also we consider the following Taylor expansion for $g(y, S, S')$:
\begin{equation}
\label{taylor:g}
g(y, S, S')= a_0+a_1y+a_2 y^2+\ldots,
\end{equation} 
where $a_i$'s are functions of $S$, $S'$. 

In addition, we assume $A_0=0$. Then by multiplying
(\ref{transport_one_dim_general:1}) and (\ref{transport_one_dim_general:2}) by
$1$, $\nu$, and/or $y$, adding or subtracting, and integrating with respect to
$y$ on $\r$, and applying the fundamental theorem of calculus and integration
by parts, we obtain the following equations for macroscopic density and flux
and their first moments:
 \begin{align}
\frac{\partial n}{\partial t}+\frac{\partial j}{\partial x}\;&=\;0, \label{moment_eq1}\\
 \frac{\partial j}{\partial t}+\nu^2\frac{\partial n}{\partial x}\;&=\; -2 a_0 j - 2 a_1 j_1-2 \sum_{k\geq2} a_k j_k,  \label{moment_eq2}\\
  \frac{\partial n_1}{\partial t}+\frac{\partial j_1}{\partial x}\;&=\; B_0 j+ A_1 n_1+B_1 j_1 +\sum_{k\geq2}A_{k} n_{k}+\sum_{k\geq2}B_{k}j_k,  \label{moment_eq3}
\\
 \frac{\partial j_1}{\partial t} + \nu^2 \frac{\partial n_1}{\partial x}\;&=\;\nu^2 B_0 n+\nu^2 B_1 n_1+( A_1-2a_0) j_1  \label{moment_eq4}
\\
& \qquad\qquad+\nu^2\sum_{k\geq2}B_{k} n_{k}+\sum_{k\geq2}(A_{k}-2a_{k-1})j_k\nonumber
 \end{align}
Note that by condition (\ref{condition_p+-}), for any $i=0, 1,
\ldots$ we have that:
\[
\displaystyle\lim_{y\to\pm\infty} y^i (p^+\pm p^-)\;=\;0.
\] 

\subsection{Parabolic scaling}
In this section, we introduce a parabolic scaling to derive a
chemotaxis equation from the moment equations
(\ref{moment_eq1})-(\ref{moment_eq4}).  
Let $L$, $T$, $\nu_0$, and $N_0$ be scale factors for the length, time,
velocity, and particle density respectively, and define the following
dimensionless parameters:
\begin{align}
\quad \hat{\nu}\;&=\;\frac{\nu}{\nu_0},\label{parabolic_scaling:1}\\
 \hat{n}\;&=\;\frac{n}{N_0}, \quad \hat{j}\;=\;\frac{j}{N_0\nu_0}, \quad
 \hat{n}_i\;=\;\frac{n_i}{N_0}, \quad \hat{j}_i\;=\;\frac{j_i}{N_0\nu_0}, \quad\mbox{for}\quad i= 1, 2, \ldots\label{parabolic_scaling:2}\\
 \hat{a}_i\;&=\;{a_i}T,\quad \hat{A}_i\;=\;{A_i}T,\quad \hat{B}_i\;=\;{B_i}L,\quad\mbox{for}\quad i= 0, 1, \ldots\label{parabolic_scaling:3}
\end{align}

The parabolic scales of space and time are given by:
\begin{align}\label{time_scaling}
  \hat{x}\;=\;\displaystyle\left(\frac{\epsilon L}{{\nu_0T}}\right)\frac{x}{L},\quad\hat{t}\;=\;\epsilon^2\frac{t}{T},
  \end{align}
  for any arbitrary $\epsilon$.
  
Now assume that under some conditions, for any $i\geq2$, the $j_i$'s and $n_i$'s
are much smaller than $j_1$ and $n_1$ and can be neglected. (For example see Lemma \ref{order_higher_of_moments} below.) 
Therefore, the dimensionless form of moment equations
(\ref{moment_eq1})-(\ref{moment_eq4}), become:
 \begin{align}
\epsilon^2\frac{\partial \hat{n}}{\partial \hat{t}}+\epsilon\frac{\partial \hat{j}}{\partial \hat{x}}\;&=\;0, \label{moment_dimensionless_eq1}\\
\epsilon^2 \frac{\partial\hat{j}}{\partial \hat{t}}+\epsilon\hat{\nu}^2\frac{\partial \hat{n}}{\partial \hat{x}}\;&=\; -2 \hat{a}_0\hat{j} - 2 \hat{a}_1 \hat{j}_1,  \label{moment_dimensionless_eq2}\\
 \epsilon^2 \frac{\partial \hat{n}_1}{\partial \hat{t}}+\epsilon\frac{\partial \hat{j}_1}{\partial \hat{x}}\;&=\; \epsilon \hat{B}_0 \hat{j}+ \hat{A}_1 \hat{n}_1+\epsilon \hat{B}_1 \hat{j}_1,  \label{moment_dimensionless_eq3}\\
\epsilon^2 \frac{\partial \hat{j}_1}{\partial \hat{t}} +\epsilon \hat{\nu}^2 \frac{\partial \hat{n}_1}{\partial \hat{x}}\;&=\;\epsilon\hat{\nu}^2 \hat{B}_0 \hat{n}+\epsilon\hat{\nu}^2 \hat{B}_1 \hat{n}_1+( \hat{A}_1-2\hat{a}_0) \hat{j}_1  \label{moment_dimensionless_eq4}
 \end{align}
 
 Next, we write Equations (\ref{moment_dimensionless_eq1})-(\ref{moment_dimensionless_eq4}) in a matrix form, as follows:
 \begin{equation}\label{matrix_form}
 \epsilon^2\frac{\partial \hat{w}}{\partial \hat{t}}+\epsilon\frac{\partial}{\partial \hat{x}} P\hat{w}\;=\; \epsilon Q \hat{w}+R \hat{w},
 \end{equation}
where $\hat{w}\;=\;\left(\hat{n}, \hat{j}, \hat{n}_1, \hat{j}_1\right)^T$ and the matrices $P$, $Q$, and $R$ defined as follows:
\[P\;=\;\left(\begin{array}{cccc}0 & 1 & 0 & 0 \\\hat{\nu}^2 & 0 & 0 & 0 \\0 & 0 & 0 & 1\\0 & 0 & \hat{\nu}^2& 0\end{array}\right),\]
\[Q\;=\;\left(\begin{array}{cccc}0 & 0 & 0 & 0 \\0 & 0 & 0 & 0 \\0 & \hat{B}_0 & 0 & \hat{B}_1 \\\hat{\nu}^2 \hat{B}_0 & 0 & \hat{\nu}^2\hat{B}_1 & 0\end{array}\right), \qquad 
R\;=\;\left(\begin{array}{cccc}0 & 0 & 0 & 0 \\0 & -2\hat{a}_0 & 0 & -2\hat{a}_1 \\0 & 0 & \hat{A}_1 & 0 \\0 & 0 & 0 &\hat{A}_1-2\hat{a}_0\end{array}\right).\] 

Assuming the regular perturbation expansion for $w$, 
\[\hat{w}\;=\;\hat{w}^0+\epsilon \hat{w}^1+\epsilon^2 \hat{w}^2+\ldots,\quad\mbox{where}\quad \hat{w}^i\;=\;\left(\hat{n}^i,\hat{ j}^i, \hat{n}_1^i, \hat{j}_1^i\right)^T,\]
and comparing the terms of equal order in $\epsilon$ in (\ref{matrix_form}), we get:
\begin{align}
\epsilon^0&:\quad R \hat{w}^0\;=\;0\quad\Rightarrow\quad \hat{w}^0\;=\;(\hat{n}^0, 0, 0, 0)^T\\
\epsilon^1&:\quad R \hat{w}^1+Q\hat{w}^0\;=\;\frac{\partial}{\partial \hat{x}} P\hat{w}^0\nonumber\\
&\Rightarrow\quad\left(\begin{array}{c}0 \\-2\hat{a}_0\hat{j}^1-2\hat{a}_1\hat{j}_1^1 \\\hat{A}_1\hat{n}_1^1 \\(\hat{A}_1-2\hat{a}_0)\hat{j}_1^1+\hat{\nu}^2\hat{B}_0\hat{n}^0\end{array}\right)\;=\;\left(\begin{array}{c}0 \\\hat{\nu}^2\frac{\partial}{\partial \hat{x}}\hat{n}^0 \\0 \\0\end{array}\right)\label{epsilon^1}
\end{align}
From Equation (\ref{epsilon^1}), we can derive the following equation for $\hat{j}_1^1$:
\[\hat{j}_1^1\;=\;-\frac{\hat{\nu}^2\hat{B}_0}{\hat{A}_1-2\hat{a}_0}\hat{n}^0,\]
and therefore, using the same Equation, we obtain the following equation for
$\hat{n}^0$ and $\hat{j}^1$:
\begin{equation}\label{eq:1}
-2\hat{a}_0\hat{j}^1+2\hat{a}_1\frac{\hat{\nu}^2\hat{B}_0}{\hat{A}_1-2\hat{a}_0}\hat{n}^0\;=\;\hat{\nu}^2\frac{\partial}{\partial \hat{x}}\hat{n}^0
\end{equation}
Now we compare the terms with order $\epsilon^2$:
\begin{equation}
\epsilon^2:\quad R \hat{w}^2\;=\;-Q\hat{w}^1+\frac{\partial}{\partial \hat{x}} P\hat{w}^1+\frac{\partial}{\partial \hat{t}} \hat{w}^0\label{epsilon^2}.\\
\end{equation}
Note that $(1, 0, 0, 0)^T$ is in the kernel of $R$ and the right hand side of (\ref{epsilon^2}) is in the image of $R$. Therefore their inner product is zero:
\begin{equation}
\frac{\partial}{\partial \hat{x}} \hat{j}^1+\frac{\partial}{\partial \hat{t}} \hat{n}^0\;=\;0\label{eq:2}.\\
\end{equation}
Equation (\ref{eq:1}) together with Equation (\ref{eq:2}) give the following
 equation for $n^0$ in the dimensionless variables: 
\begin{equation}\label{n0_equation_dimensionless}
\frac{\partial \hat{n}^0}{\partial \hat{t}}\;=\; \frac{\hat{\nu}^2}{2\hat{a}_0}\frac{\partial^2\hat{n}^0}{\partial \hat{x}^2}-\frac{\hat{a}_1\hat{B}_0\hat{\nu}^2}{\hat{a}_0(\hat{A}_1-2\hat{a}_0)}\frac{\partial \hat{n}^0}{\partial \hat{x}}.
\end{equation}
Since $n(x,t)=n^0(x,t)+\mathcal{O}(\epsilon)$, if we neglect the $\mathcal{O}(\epsilon)$ term, Equation (\ref{n0_equation_dimensionless}) leads to the following chemotaxis equation in dimensionless variables: 
\begin{equation}\label{main_equation_dimensionless}
\frac{\partial \hat{n}}{\partial \hat{t}}\;=\; \frac{\hat{\nu}^2}{2\hat{a}_0}\frac{\partial^2\hat{n}}{\partial \hat{x}^2}-\frac{\hat{a}_1\hat{B}_0\hat{\nu}^2}{\hat{a}_0(\hat{A}_1-2\hat{a}_0)}\frac{\partial \hat{n}}{\partial \hat{x}}.
\end{equation}
 Changing back to the original (dimensional) variables, we obtain the following PDE:
\begin{equation}\label{main_equation}
\frac{\partial n}{\partial t}\;=\; \frac{\nu^2}{2a_0}\frac{\partial^2n}{\partial x^2}-\frac{a_1B_0\nu^2}{a_0(A_1-2a_0)}\frac{\partial n}{\partial x}.
\end{equation}

\section{Example}

In this example, we assume the internal state evolves according to the
following ODE system:

\begin{equation}
\label{deterministic_model_example}
\frac{dy}{dt}\;=\;py (q- a),
\end{equation}
 where $a=\displaystyle\frac{1}{1+K(\frac{S}{y})^N}$, and $p$, $q$, $K$, and $N$ are positive constants.   
 
 This system provides a simple model of chemotactic behavior in \emph{E.\ coli}
bacteria, as discussed below in the section on simulations.

 By ignoring the tumbling time, we consider the following equation for the tumbling rate:
 \begin{equation}
 \label{response}
 \lambda(y)\;=\;r a^H,
\end{equation}
where $r$ and $H$ are positive constants. 

The objective is to derive a parabolic equation for the macroscopic density function.

 It is convenient to define a new internal state variable as follows:
 \begin{equation}
 \label{new_var}
 z\;=\;p(a-q).
\end{equation}
 
A simple calculation shows that  
\begin{align}
\frac{dz}{dt}\;&=\;\frac{N}{p}z(z+pq)(z+pq-p)\pm\nu \frac{N}{p}\displaystyle\frac{S'}{S}(z+pq)(z+pq-p)\label{new_system:1}\\
\lambda(z)\;&=\; \frac{r}{p^H}(z+pq)^H\label{new_system:2}
\end{align}

\begin{lemma}[shallow condition]
\label{shallow_condition_example}
Let $c=\min\{pq, p-pq\}$. If
\[\left|{\displaystyle\frac{S'}{S}}\right|\leq \frac{c}{\nu} \quad\mbox{and}\quad \abs{z(0)}\leq c,\]
 then $\abs{z(t)}\leq c$ for all $t\geq0$.
\end{lemma}

\begin{proof}
At $z(0)=c$, since $\left|{\displaystyle\frac{S'}{S}}\right|\leq \displaystyle\frac{c}{\nu}$, $z\pm\nu\displaystyle\frac{S'}{S}\geq0$. By the definition of $c$, $z+pq-p\leq0$, while $z+pq\geq0$. Therefore at $z=c$, $\displaystyle\frac{dz}{dt}\leq0$. On the other hand, at $z(0)=-c$, $z\pm\nu\displaystyle\frac{S'}{S}\leq0$, $z+pq-p\leq0$, and $z+pq\geq0$. Therefore at $z=-c$, $\displaystyle\frac{dz}{dt}\geq0$. Hence for any $t\geq0$, $\abs{z(t)}\leq c$.
\end{proof}

We'll show that under the shallow condition, Lemma \ref{shallow_condition_example}, and the following parabolic dimensionless parameters, the higher macroscopic moments $j_2, j_3, \ldots$, and $n_2, n_3, \ldots$ can be ignored. 

\begin{lemma}\label{order_higher_of_moments}
As before, let $L$, $T$, $\nu_0$, and $N_0$ be scale factors for the length, time, velocity, and particle density respectively, and define the following dimensionless quantities: 
\[\widehat{\left(\frac{S'}{S}\right)}\;=\;\frac{\nu_0}{\epsilon}\frac{S'}{S}, \quad \hat{N}\;=\; TN,\quad \hat{p}\;=\; p\quad \hat{q}\;=\; q \quad\mbox{and}\quad \hat{r}\;=\; Tr.\]
All other parameters remain the same as in Equations (\ref{parabolic_scaling:1})-(\ref{parabolic_scaling:2}), and Equation (\ref{time_scaling}). 
Then under the condition of Lemma \ref{shallow_condition_example}, for any $i\geq1$, 
\[\frac{\hat{j}_i}{\hat{n}}\;\leq\;\mathcal{C}_i\epsilon^i\quad\mbox{and}\quad \frac{\hat{n}_i}{\hat{n}}\;\leq\;\mathcal{D}_i\epsilon^i,\]
for some constants $\mathcal{C}_i\;=\;\mathcal{O}(1)$, and $\mathcal{D}_i\;=\;\mathcal{O}(1)$.
\end{lemma}

\begin{proof}
Note that $\left|{\displaystyle\frac{S'}{S}}\right|\;\leq\; K$ implies 
$\left|\widehat{\left(\displaystyle\frac{S'}{S}\right)}\right|\;\leq \;\displaystyle\frac{\nu_0}{\epsilon}K$. Hence $K\;=\;\displaystyle\frac{\epsilon}{\nu_0}\bar{K}$, where $\bar{K}\;=\;\mathcal{O}(1)$. Now by Lemma \ref{shallow_condition_example}, we have $\abs{z}\;\leq\;\displaystyle\nu\frac{\epsilon}{\nu_0}\bar{K}\;=\;\epsilon\hat{\nu}\bar{K}$.
\[
\begin{array}{lcl}
\hat{j}_i &=& \displaystyle\frac{j_i}{\nu_0 N_0} \\
&=&\displaystyle\frac{\nu}{\nu_0 N_0}   \int_{\r} z^i(p^+-p^-) \;dz\\
&\leq&\displaystyle\frac{\hat\nu}{N_0}   \int_{\r} z^i(p^++p^-) \;dz\quad(=\hat\nu \hat{n}_i)\\
&\leq&\displaystyle\frac{\hat{\nu}}{ N_0}(\epsilon\bar K\hat{\nu})^i \int_{\r}  (p^++p^-) \;dz\\
&=&\epsilon^i\mathcal{C}_i\hat{n},
\end{array}
\]
where $\mathcal{C}_i\;=\;{\bar{K}}^i\hat{\nu}^{i+1}\;=\;\mathcal{O}(1)$. Note that $\displaystyle\frac{\hat{n}_i}{\hat{n}}\;\leq\;\mathcal{D}_i$, where 
$\mathcal{D}_i\;=\;\displaystyle\frac{\mathcal{C}_i}{\hat\nu}.$
\end{proof}

Using the notations of Equations (\ref{taylor:f0})-(\ref{taylor:f1}), 
\begin{align}
A_0\;=\;0, \quad A_1\;=\;Npq \left(q-1\right), \quad B_0\;&=\; N\frac{S'}{S}pq\left(q-1\right),
\end{align}
and by notation of Equation (\ref{taylor:g}), the first two coefficients of the Taylor expansion of $g$ is as follows: 
\begin{align}
a_0\;=\;rq^H, \quad a_1\;=\;\displaystyle\frac{rHq^H}{pq}.
\end{align}

Therefore, using the Equation (\ref{main_equation}), the chemotaxis equation for this particular example is:
\begin{equation}\label{main_equation_example}
\frac{\partial n}{\partial t}\;=\;\frac{\nu^2}{2r}\left(\frac{1}{q}\right)^H\frac{\partial^2 n}{\partial x^2}-\frac{NH\left(q-1\right)\nu^2}{Npq\left(q-1\right)-2rq^H}\frac{S'}{S}\frac{\partial n}{\partial x}.
\end{equation}

 \section{Simulations}

In this section, we provide agent-based simulations of the full jump-Markov
system, and comparisons with the parabolic model, for systems of the special
form in \eqref{deterministic_model_example}, which we repeat here for
convenience: 
\[
\frac{dy}{dt} \;=\; p  y  (q- a) 
\,,\quad\quad
a \;=\;  \frac{1}{1+K(\frac{S}{y})^N} \,.
\]
The jump (or ``tumbling'' for bacteria) rate has the form $\lambda (y)\;=\;r a^H$
in~\eqref{response}.
The parameters $p$, $q$, $K$, $N$, $r$, and $H$ are all positive.

For ligand concentrations $K_I$$\ll$$S$$\ll$$K_A$, where $K_I\approx 18.2\mu M$
and $K_A\approx3000\mu M$ are the dissociation constants for inactive and active
Tar receptors respectively, the above equations provide a simple but phenomenologically accurate
model%
\footnote{For convenience of analysis, we are using $y = e^{\alpha m}$ as a state
variable, instead of the methylation level ``$m$'' as done in other papers.}
of the chemotactic response of \emph{E.\ coli} bacteria to MeAsp; see for
example \cite{TuShimizuBerg2008}, \cite{Jiang_PLOS2010}.
Furthermore, this is the range in which~\cite{shoval10} predicted,
and~\cite{stocker_shimizu11} experimentally verified, scale-invariant behavior
for \emph{E.\ coli} responses to MeAsp.  (See also
\cite{shoval_alon_sontag_2011} for further theory of scale invariance.)
To stay in this range, we use ligand concentrations very close to, and mostly
larger than, $S = 100\mu M$. 
Since our objective is to understand the quality of the parabolic
(reaction-diffusion) equation, we depart slightly from models cited above, in
postulating an instantaneous re-orientation after tumbling.  A model with
tumbling would require additional analysis. Also, since motion is one- dimensional, we ignore rotational random drifts from linear movement. 

The parameters in previous studies, see for example~\cite{Jiang_PLOS2010}, are
as follows:%
\footnote{In terms of the parameters used in~\cite{Jiang_PLOS2010},
            $p=\alpha (k_R+k_B)$, where $\alpha =1.7$ and $k_R=k_B=0.005$.}
\[
\nu =0.0165, \;\;  
N=6, \;\;
q=0.5, \;\;
r = 1280, \;\;
H=10, \;\;
K = 0.000740, \;\;
p = 0.017
\]
(in appropriate units corresponding to $\mu M$ concentrations, times in
seconds, and lengths in millimeters).
We start with these, but we will vary $p$ in order to understand how
the speed of adaptation (i.e, the time-scale at which the state variable $y$
evolves) affects the quality of our theoretical predictions.

In our numerical experiments, we take a one-dimensional channel of length 10
(in units of millimeters), and start all agents (cells) in the middle
position, $x=5$, randomizing the initial direction of movement as right or
left with probability $1/2$.
The initial level $y(0)$ of every cell is picked such that the activity $a$
equals the adapted value $q$.  (This is in accordance with the pre-adaptation
setup in the microfluidics experiments in~\cite{stocker_shimizu11}.)

The length is picked large enough so that, in the time intervals considered
(up to $[0,200]$), no boundaries are reached, so that, for all practical
purposes, we are working on an infinite domain.
We always take 100,000 cells, and display histograms based on 100 equal-sized
bins.  (These numbers represent a heuristic compromise between computational
effort and smoothness of empirical densities.)

\newcommand{\slope}{\rho }

We only consider exponential gradients $S(x)=\kappa  e^{\slope x}$, in which case
the advection term does not depend on $x$, because $S'(x)/S(x)=\slope$ is
constant.  We call $\slope$ the ``slope'' (more precisely, this is the slope of
$\log S$), and the shallow-gradient condition amounts to requiring $\slope\ll1$.
We typically pick $\slope$ in the range 0.1 to 1.
Under the assumption that $S(x)$ is exponential, our solution
\eqref{main_equation_example} becomes a constant coefficient
advection-diffusion equation:
\[
\frac{\partial n}{\partial t}\;=\; 
D
\frac{\partial^2n}{\partial x^2}
-
V
\frac{\partial n}{\partial x}.
\]
The general solution of this equation, when starting from a Dirac delta
function at position $x_0(=5)$, has the form
\[
n(x,t) = {\cal N}(x+x_0 - Vt,2Dt)
\]
where ${\cal N}(\mu ,\sigma )$ is a Gaussian density with mean $\mu $ and variance
$\sigma $. 
In other words, the solution is a translate of the fundamental solution of the
heat equation.
For purposes of comparison, the densities displayed at any given time $t$
are plotted together with this theoretical prediction.
As remarked earlier, the constant $\kappa $ is picked so that $\kappa S(x_0)=100$,
that is, $\kappa  = 100 e^{-5\slope}$.

\subsection{Slope $\slope=0.1$, various values of $p$}

In Figures \ref{0.1.0.017} to \ref{0.1.1}, we display simulated and
theoretical distributions at times $t=10,100,200$ as well as a plot of the
means of the distribution on the interval $[0,200]$.  Agreement to theory is
very good, with larger $p$ (faster internal adaptation dynamics) leading to
closer fits.

\newcommand{\ns}{\null\vspace*{-10pt}}
\abovecaptionskip-20pt

\ns
\begin{figure}[ht]
\picc{0.2}{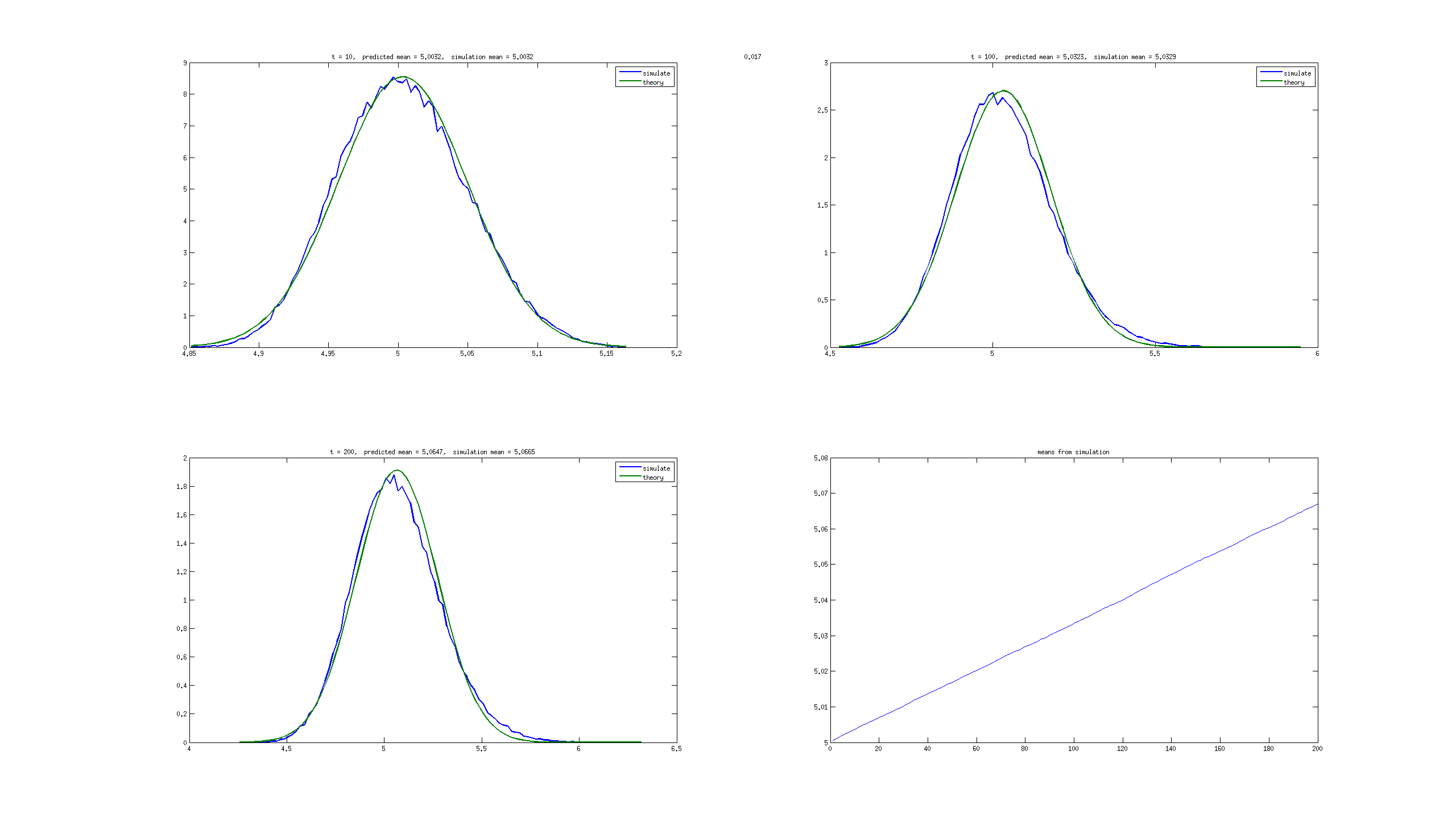}
\ns
\caption{$\slope=0.1,p=0.017$}
\label{0.1.0.017}
\end{figure}

\ns
\begin{figure}[ht]
\picc{0.2}{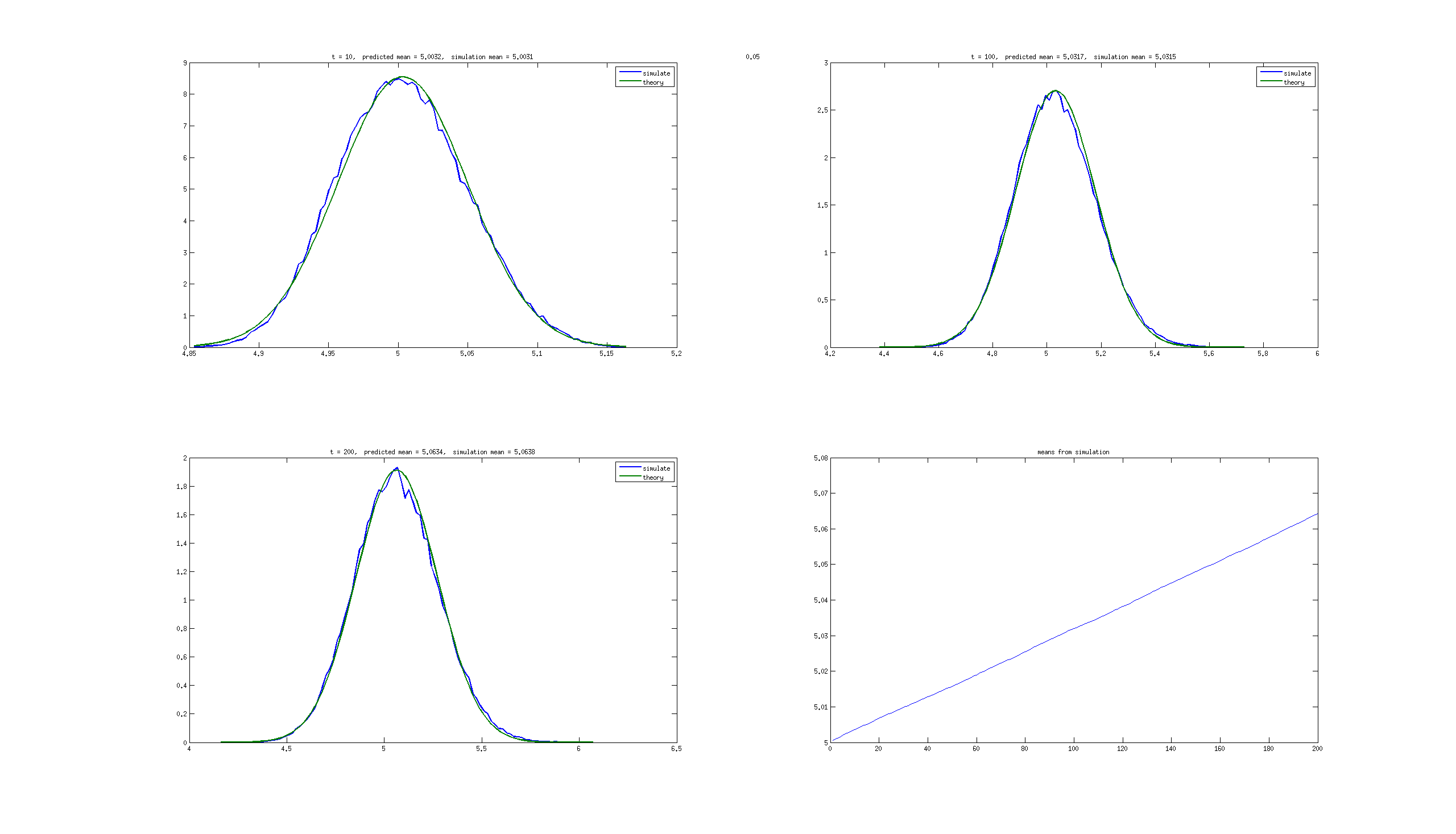}
\caption{$\slope=0.1,p=0.05$}
\end{figure}

\ns
\begin{figure}[ht]
\picc{0.2}{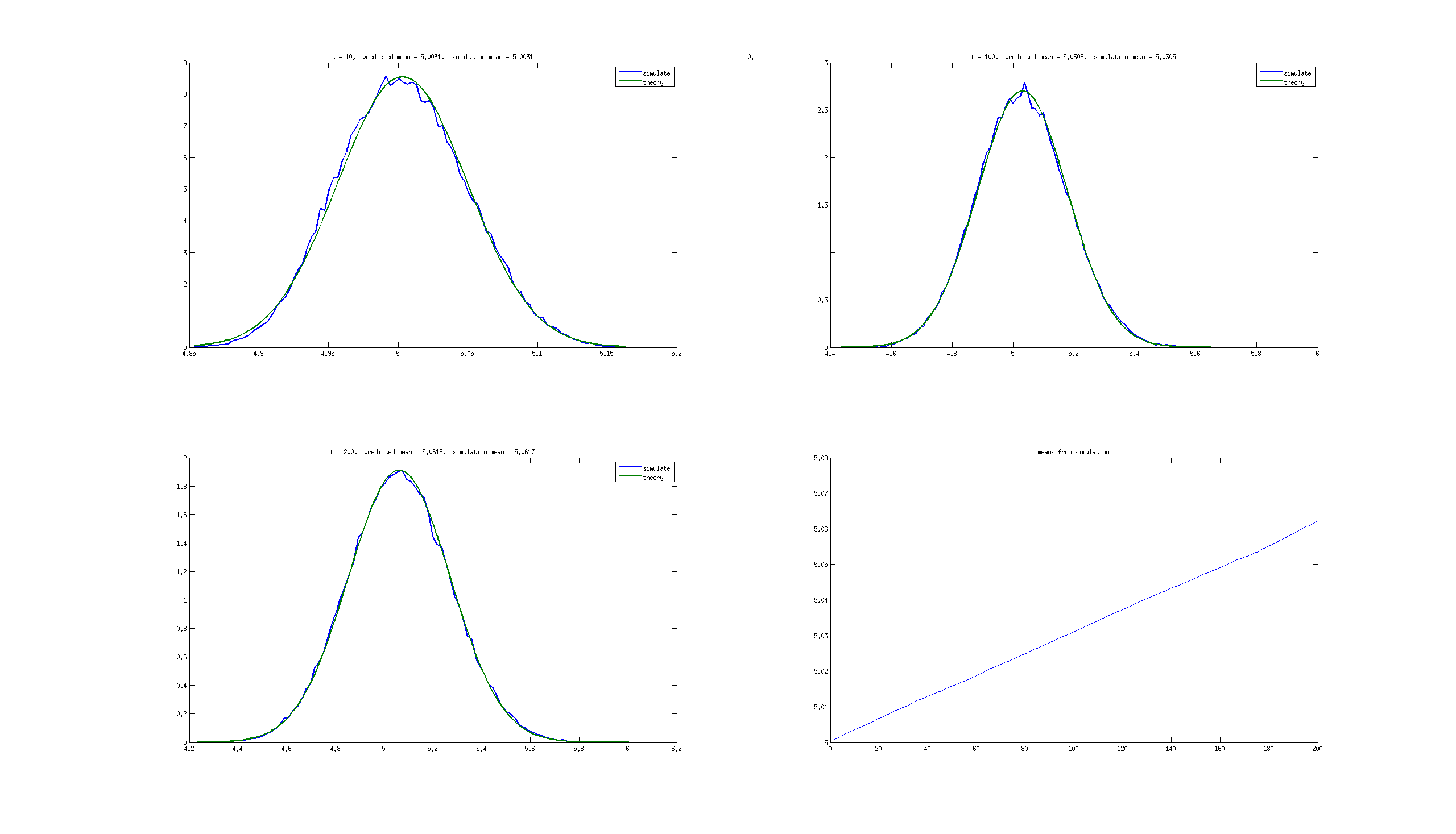}
\caption{$\slope=0.1,p=0.1$}
\end{figure}

\ns
\begin{figure}[ht]
\picc{0.2}{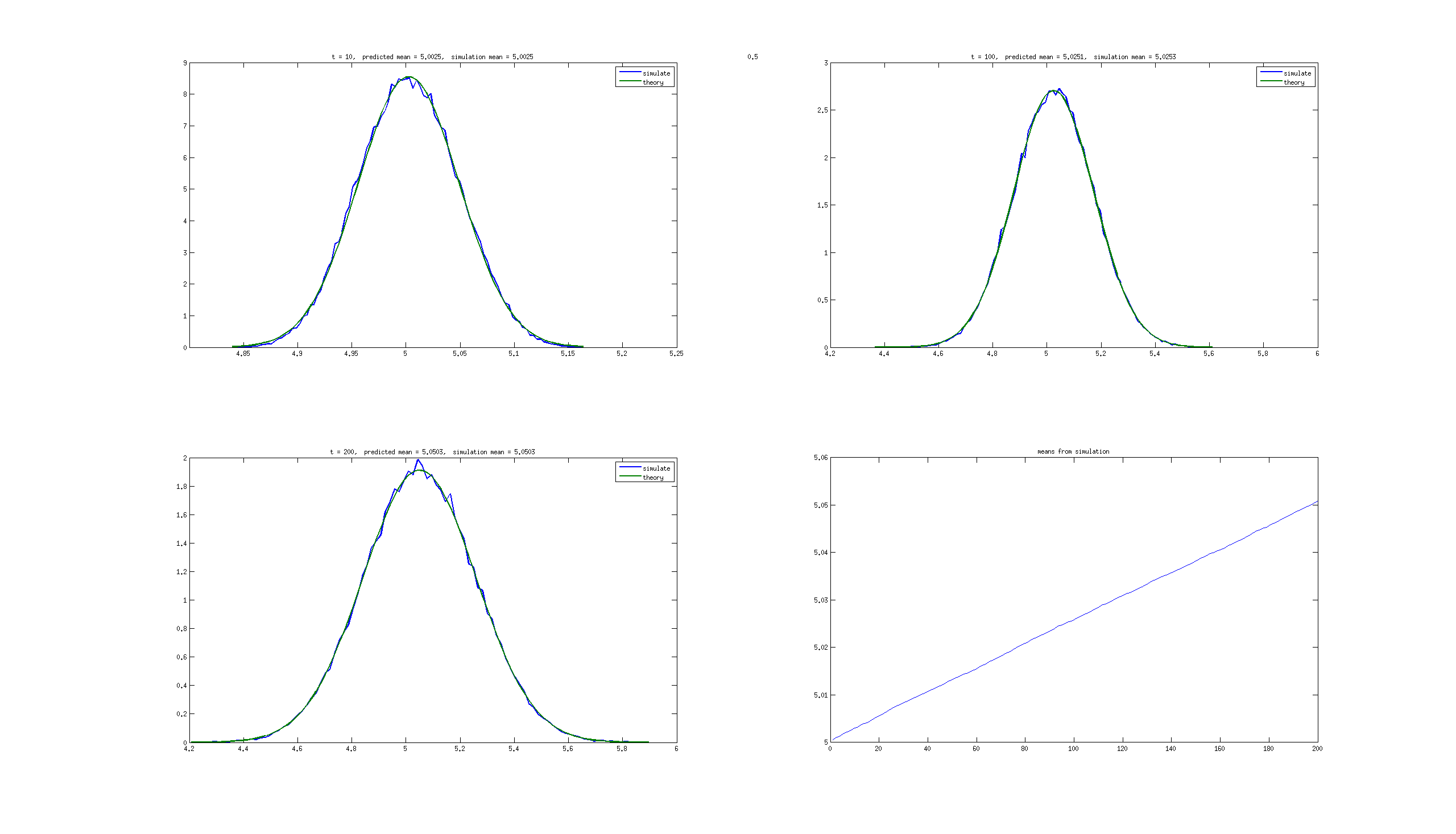}
\caption{$\slope=0.1,p=0.5$}
\end{figure}

\ns
\begin{figure}[ht]
\picc{0.2}{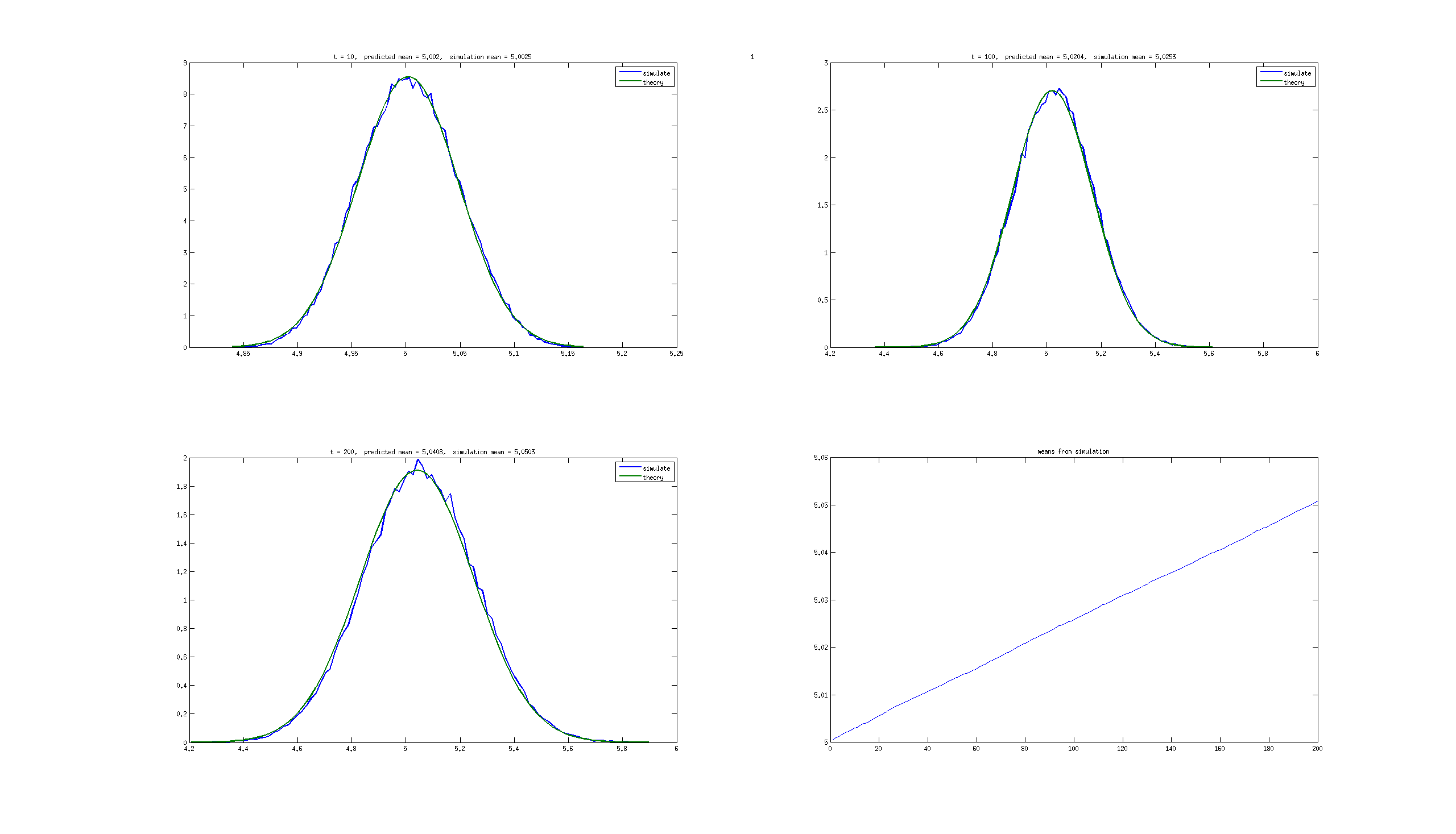}
\caption{$\slope=0.1,p=1$}
\label{0.1.1}
\end{figure}

\clearpage
\subsection{Slope $\slope=1$, various values of $p$}

In Figures \ref{1.0.017} to \ref{1.1}, we give plots with simulated and
theoretical distributions at times $t=10,100,200$ as well as a plot of the
means of the distribution on the interval $[0,200]$.  Agreement to theory is
now poor when $p=0.017$ and $p=0.1$, but is considerably better with larger
$p$ (faster internal adaptation dynamics).

\ns
\begin{figure}[ht]
\picc{0.2}{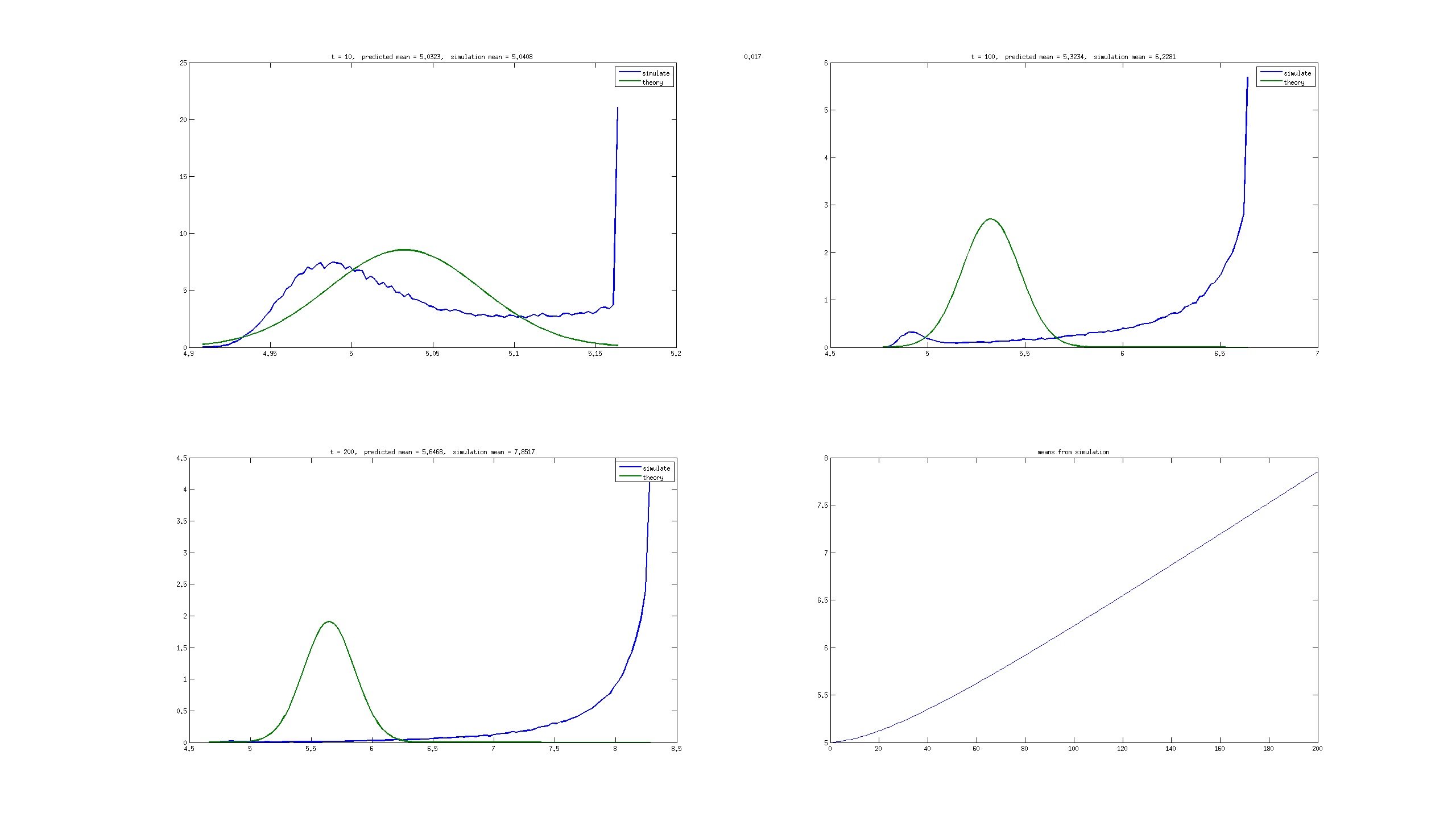}
\ns
\caption{$\slope=1,p=0.017$}
\label{1.0.017}
\end{figure}


\ns
\begin{figure}[ht]
\picc{0.2}{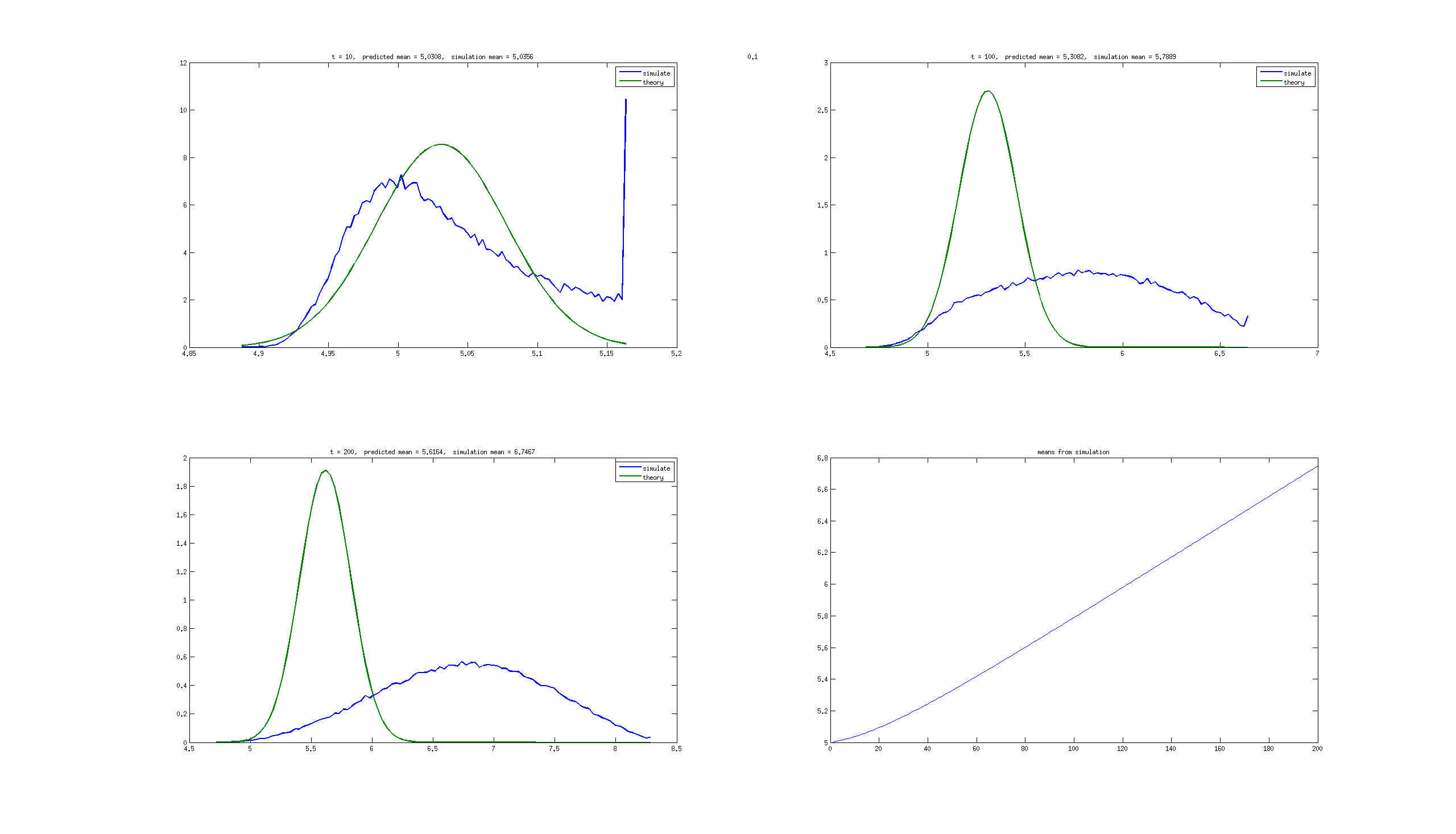}
\caption{$\slope=1,p=0.1$}
\end{figure}

\ns
\begin{figure}[ht]
\picc{0.2}{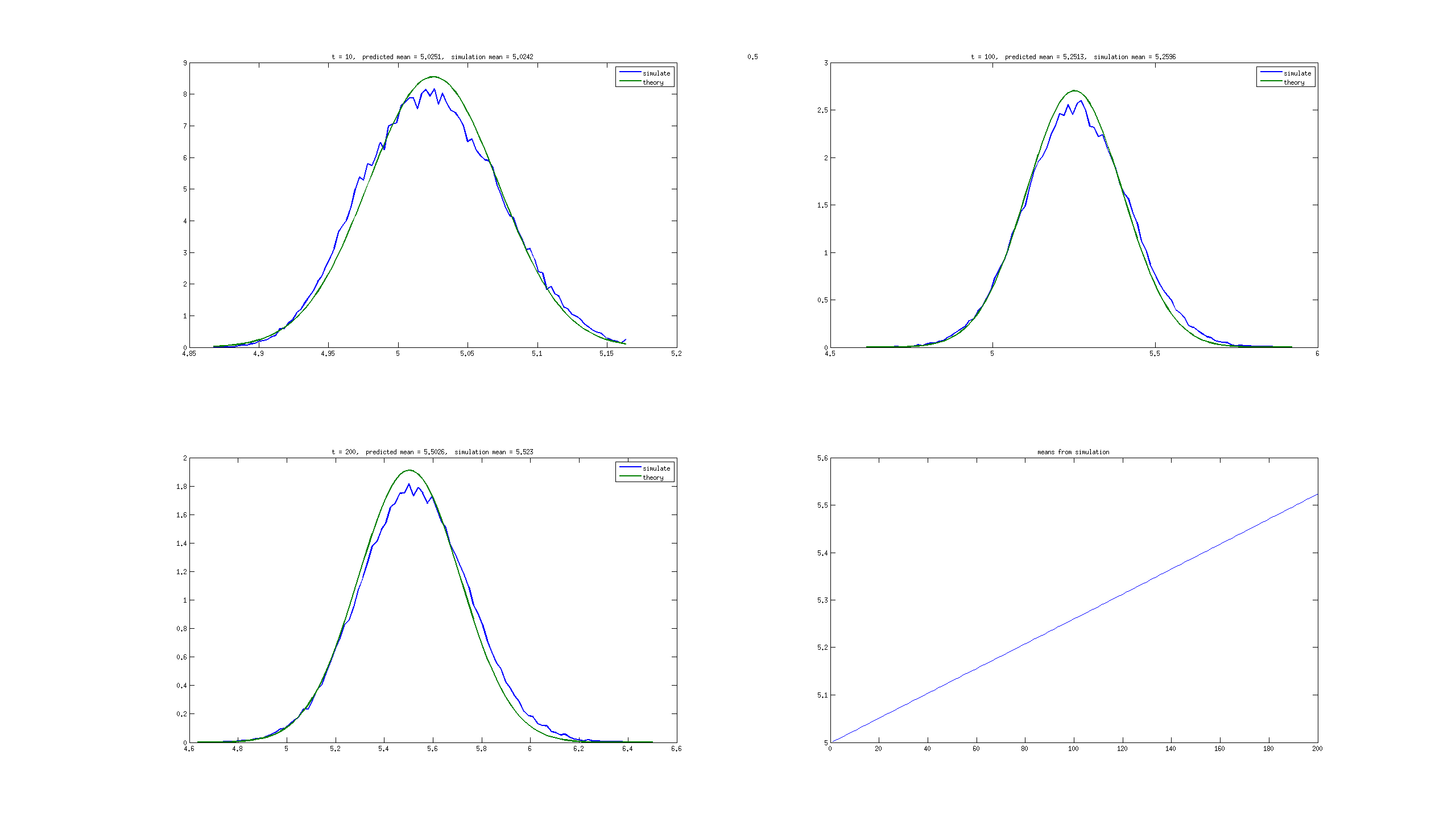}
\caption{$\slope=1,p=0.5$}
\end{figure}

\ns
\begin{figure}[ht]
\picc{0.2}{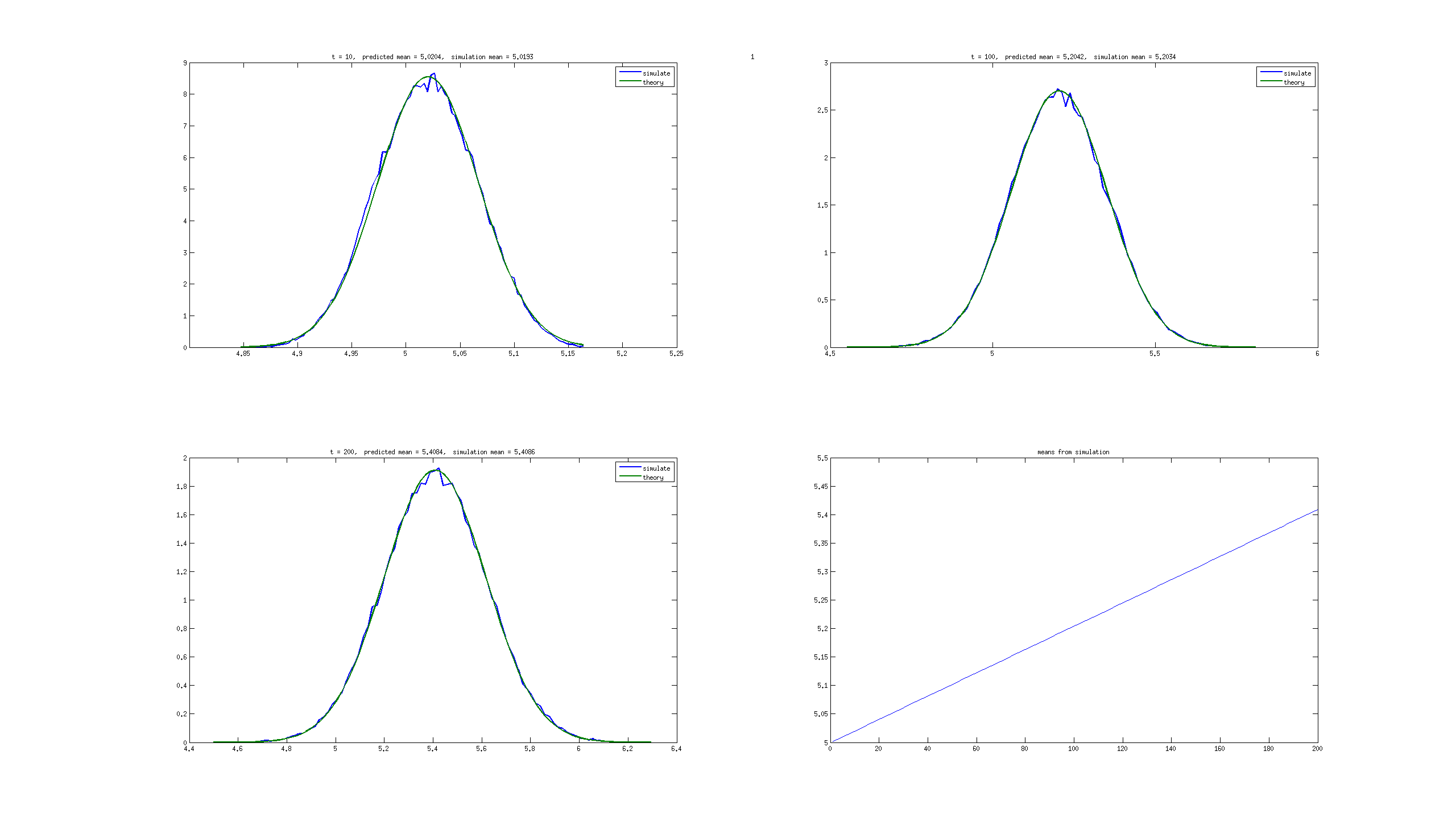}
\caption{$\slope=1,p=1$}
\label{1.1}
\end{figure}

Observe that for low values of $p$, a large number of agents (cells) appear to
be at the same ``forward'' position, moving to the right.
An explanation of this phenomenon is that the probability of tumbling in each
agent is very low.  (The second peak is explained by the fact that 1/2 of the
cells where randomized to starting in a leftward motion.  Thus, it takes a
certain time for these cells to tumble and start moving right, toward higher
nutrient concentrations.)

\clearpage
\subsection{Slopes $\slope=0.1$ to $0.5$, $p=0.017$}

To further understand the behavior when $p=0.017$, which matches theory well
when $\slope=0.1$, but badly when $\slope=1$, we show in
Figures \ref{0.1.0.017s} to \ref{0.5.0.017}
similar graphs for values $\slope=0.1,0.2,0.3,0.4,0.5$.
We display distributions at times $t=20,30,40$.
Asymmetry becomes more obvious for larger slope, and a small ``right-moving
front'' can be seen arising at $\slope=0.5$.

\ns
\begin{figure}[ht]
\picc{0.2}{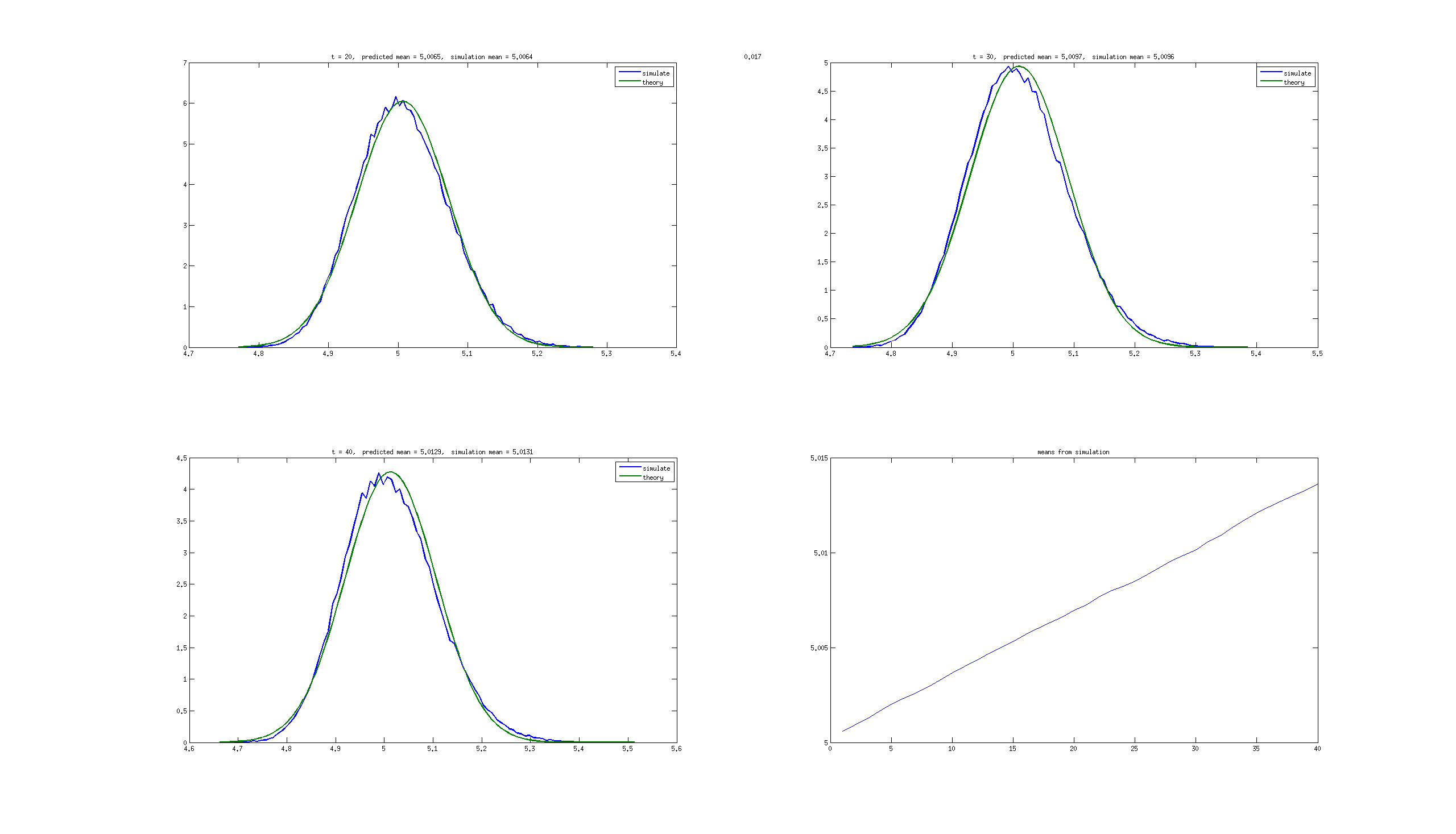}
\ns
\caption{$\slope=0.1,p=0.017$}
\label{0.1.0.017s} 
\end{figure}

\ns
\begin{figure}[ht]
\picc{0.2}{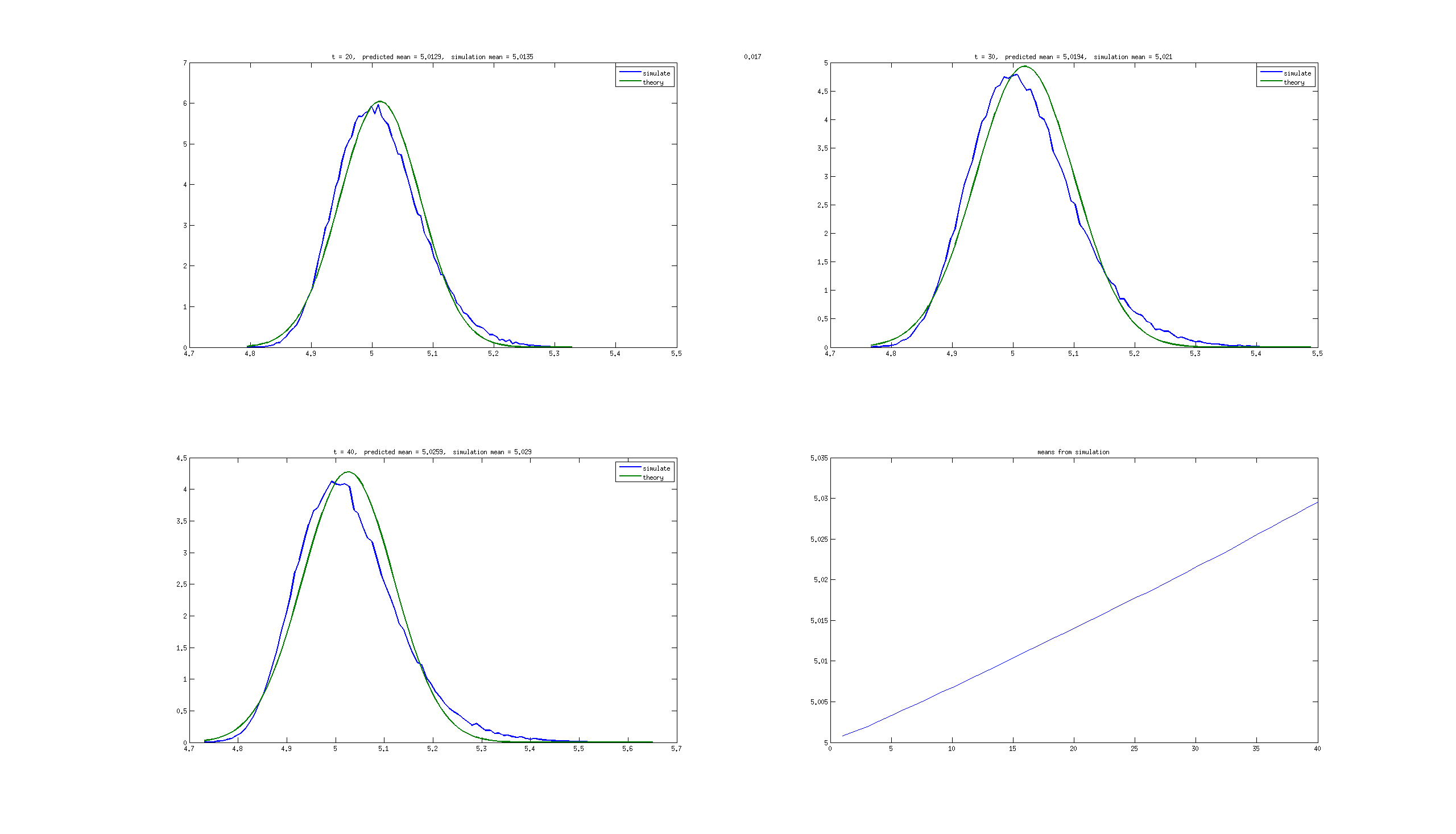}
\ns
\caption{$\slope=0.2,p=0.017$}
\label{0.2.0.017}
\end{figure}

\ns
\begin{figure}[ht]
\picc{0.2}{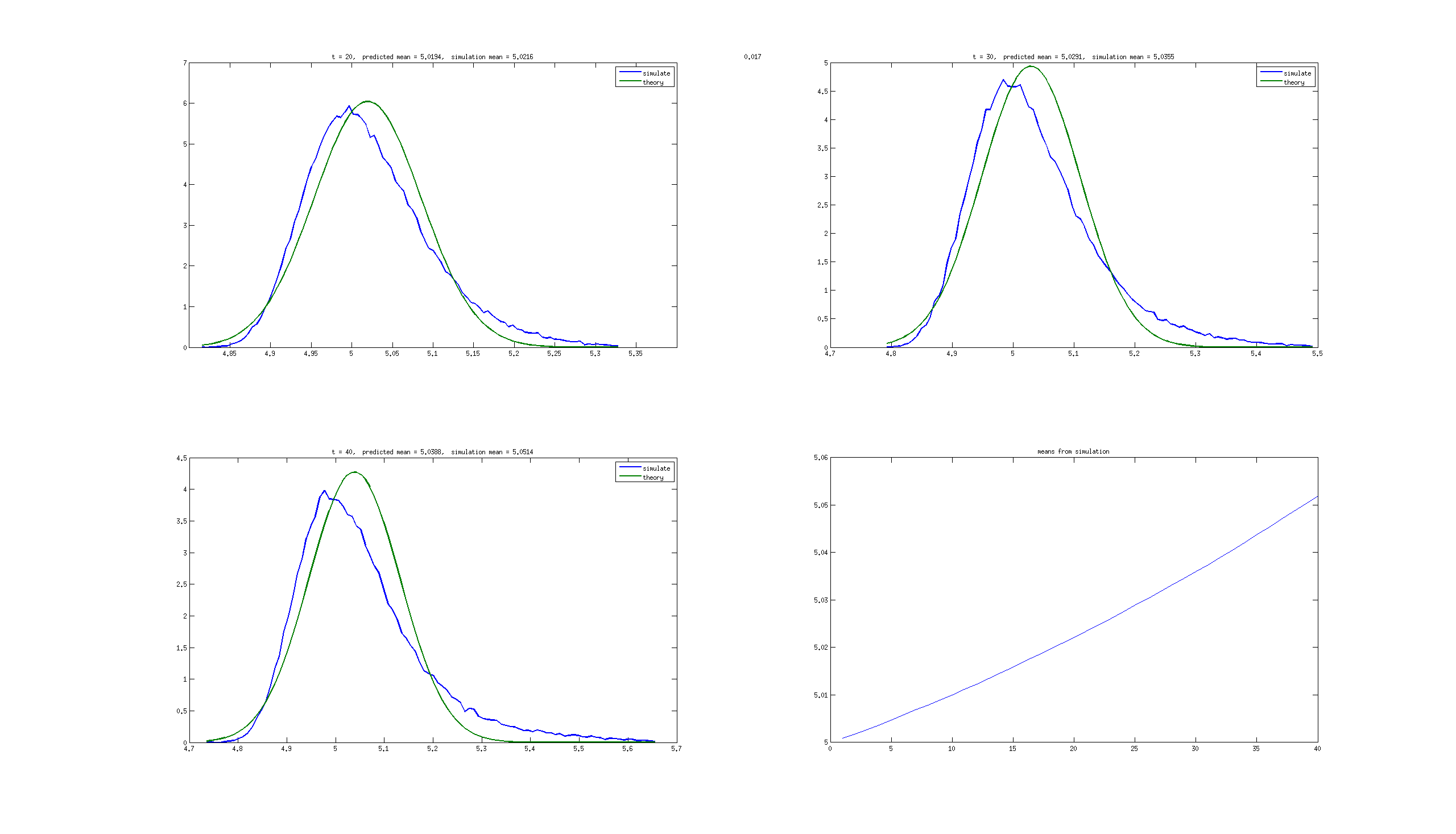}
\ns
\caption{$\slope=0.3,p=0.017$}
\label{0.3.0.017}
\end{figure}

\ns
\begin{figure}[ht]
\picc{0.2}{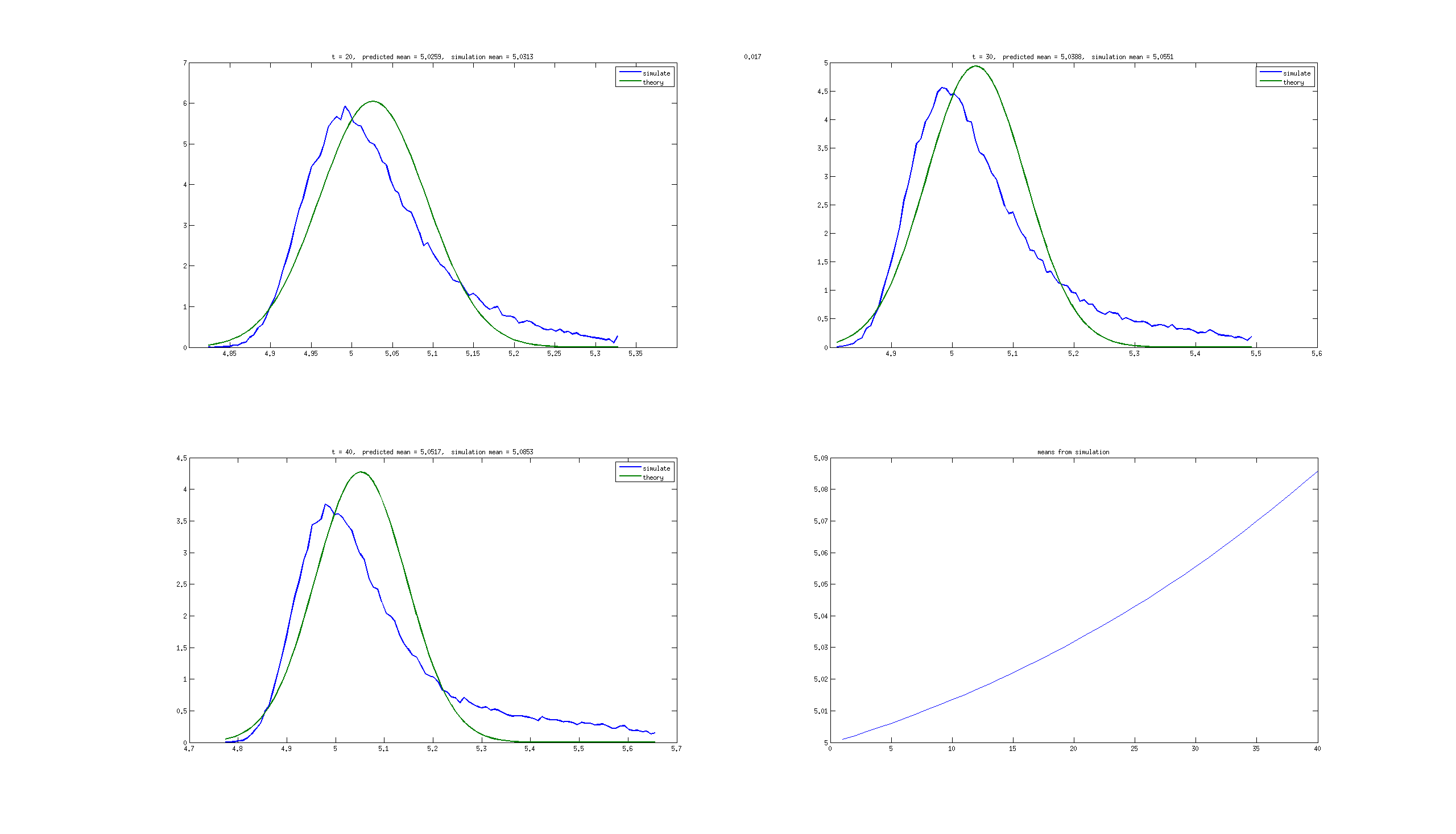}
\ns
\caption{$\slope=0.4,p=0.017$}
\label{0.4.0.017}
\end{figure}

\ns
\begin{figure}[ht]
\picc{0.2}{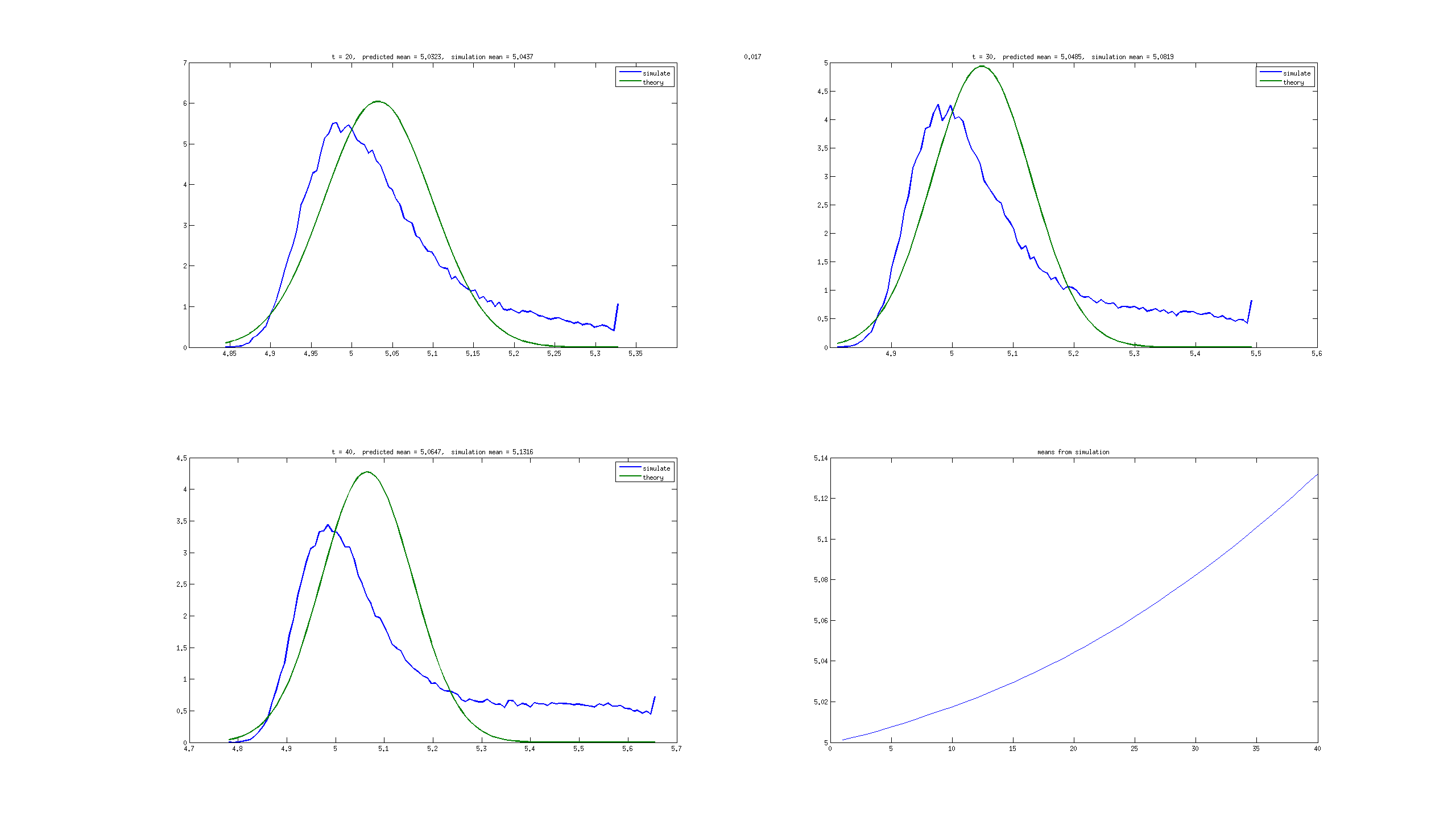}
\ns
\caption{$\slope=0.5,p=0.017$}
\label{0.5.0.017}
\end{figure}

\subsection{Slope $\slope=1$, $p=0.017,0.1,1$: plots of jump rates}

As remarked above, one may expect that the ``right-moving front'' observed for
large slopes $\slope$ and small adaptation time $p$ is due to the jump
(tumbling) rate $\lambda (t)$ (that is, $\lambda (t)dt$ the probability of a jump in an
interval $[t,t+\delta t]$) being low when $p$ is small.  This is indeed seen in the
simulations.  Figures \ref{prob0.017}, \ref{prob01}, and \ref{prob1} show the
mean value of the jump rate $\lambda (t)$, averaged over all 100,000 cells in the
simulation.  Observe that the value approaches zero when $p$ is small.
However, for larger $p$, for example $p=1$, these probabilities rapidly
approach a more or else constant (and larger) value.  These simulations are
performed on the interval $[0,100]$.

The mean and standard deviations of $\lambda (100)$ are, respectively, as follows:
\[
{\mathbf 0.017}: 0.1286, 0.7404; \;
{\mathbf 0.1}:   0.4930, 0.7413; \;
{\mathbf 1}: 1.2232, 0.2499 \,.
\]

\ns
\begin{figure}[ht]
\picc{0.4}{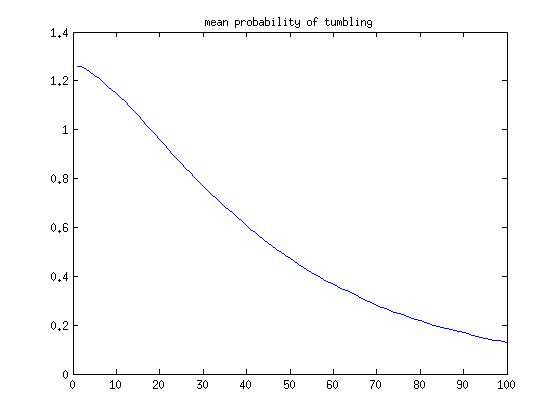}
\ns
\caption{$\slope=1$, $p=0.017$, population means of $\lambda (t)$}
\label{prob0.017}
\end{figure}

\ns
\begin{figure}[ht]
\picc{0.4}{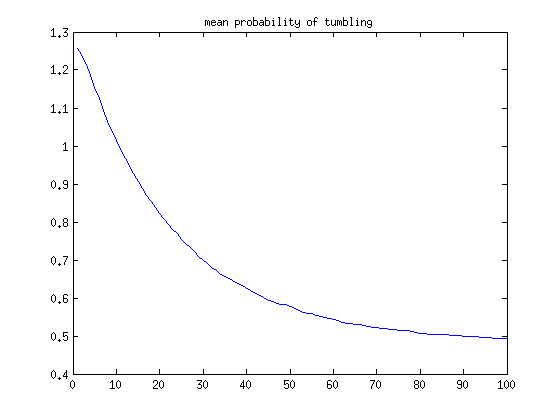}
\ns
\caption{$\slope=1$, $p=0.1$, population means of $\lambda (t)$}
\label{prob01}
\end{figure}

\ns
\begin{figure}[ht]
\picc{0.4}{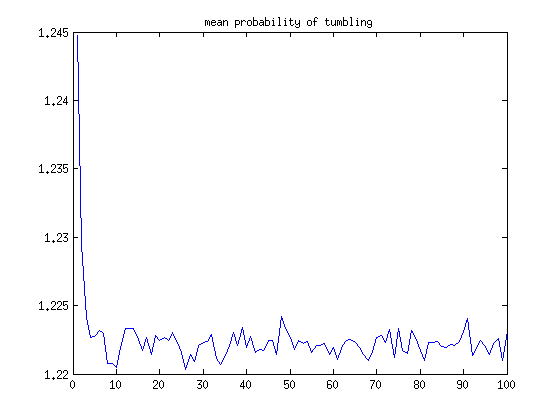}
\ns
\caption{$\slope=1$, $p=1$, population means of $\lambda (t)$}
\label{prob1}
\end{figure}

\clearpage
\newpage
\subsection{Ideal steady-state values of $\lambda $}

When $S'/S\equiv \rho $, the $z=p(a-q)$ variable in our theoretical derivation
evolves according to the following cubic differential equation:
\be{eq:cubic_z}
\frac{dz}{dt}\;=\;
\frac{N}{p}(z+pq)(z+pq-p)(z \pm \nu \rho )
\ee
where the $+$ sign is picked when the agents are moving rightward and the
$-$ sign is picked otherwise.
Physically, one is interested in solutions with positive $y(t)$, so that the
activity $a(t)$ is always in the interval $(0,1)$, which means, in terms of
the $z$ variable, that we must study Equation~\eqref{eq:cubic_z}
on the interval $J = (-pq,-pq+p)$.
We will always assume that $0<q<1$, since otherwise the system has no
equilibria for constant inputs (and in particular, does not perfectly adapt to
step signals).
Observe that $J$ is forward-invariant, since $\dot z=0$ when $z=-pq$ and
when $z=p-pq$.
There is a third root of the cubic at $z=\mp \nu \rho $, and this root belongs 
to the interval $J$ if and only if $pq-p < \pm  \nu \rho  < pq$.

Suppose now that $\rho >0$ (source of nutrient is to the right), as in our
simulations, and consider an agent that is moving rightward (``$+$'' sign).
We consider three cases:
\bi
\item
$p(q-1) < \nu \rho  < pq$:
in this case, as long as there are no jumps (direction reversals),
$z(t) \rightarrow  {\bar z} = -\nu \rho $ as $t\rightarrow \infty $.
\item
$pq < \nu \rho $:
in this case, as long as there are no jumps,
$z(t) \rightarrow  -pq$  as $t\rightarrow \infty $.
\item
$\nu \rho <p(q-1)$:
this case cannot happen, because $q<1$ and $\nu \rho >0$.
\ei
Thus, there is a bifurcation when the parameters satisfy:
\[
pq\;=\;\nu \rho \,.
\]
In terms of the exponential rate for jumps
$\lambda (t) = \frac{r}{p^H}(z(t)+pq)^H$,
we have in the first case that
\[
\lambda (t) \rightarrow  {\bar \lambda} \;=\; \frac{r}{p^H}(pq-\nu \rho )^H
\]
as $t\rightarrow \infty $, and in the second case that
\[
\lambda (t) \rightarrow  0 \,.
\]
Now, suppose that an individual agent has spent enough time moving to the right
that $z$ has achieved a value close to its steady state.
If the parameters are such that $pq<\nu \rho $ or if $pq\approx\nu \rho $,
then the rate $\lambda \approx 0$, and
there will be no further jumps in direction; the agent will continue traveling
rightward forever.
With these parameters, we will observe a front moving rightward.
For example, for $p=0.017$, $q=0.5$, and $\nu =0.0165$, this phenomenon will
happen when the slope is larger than approximately $0.5$, which is
perfectly consistent with our simulations.
On the other hand, for larger $p$, for example $p=1$, this will not happen
until the slope is very large (larger than about $30$).
In summary, for either faster dynamics ($p$ larger) or smaller slopes (smaller
$\slope$), we expect our diffusion approximation to be more accurate.
This is consistent with the simulation results.



\begin{thebibliography}{1}

\bibitem{Erban_Othmer_2004}
R.~Erban and H.~G. Othmer.
\newblock From individual to collectibe behavior in bacterial chemotaxis.
\newblock {\em SIAM Journal on Applied Mathematics}, 2004.

\bibitem{Grunbaum_2000}
D.~Grunbaum.
\newblock Advection-diffusion equations for internal state-mediated random
  walks.
\newblock {\em SIAM Journal on Applied Mathematics}, 61(1):43--73, 2000.

\bibitem{Jiang_PLOS2010}
L.~Jiang, Q.~Ouyang, and Y.~Tu.
\newblock Quantitative modeling of escherichia coli chemotactic quantitative
  modeling of escherichia coli chemotactic motion in environments varying in
  space and time.
\newblock {\em PLoS Computational Biology}, 2010.

\bibitem{stocker_shimizu11}
M.~D. Lazova, T.~Ahmed, D.~Bellomo, R.~Stocker, and T.~S. Shimizu.
\newblock {{R}esponse rescaling in bacterial chemotaxis}.
\newblock {\em Proc. Natl. Acad. Sci. U.S.A.}, 108:13870--13875, 2011.

\bibitem{11cdc_shoval_alon_sontag}
O.~Shoval, U.~Alon, and E.D. Sontag.
\newblock Input symmetry invariance, and applications to biological systems.
\newblock {\em Proc. IEEE Conf. Decision and Control, Orlando, Dec. 2011, IEEE
  Publications}, page TuA02.5, 2011.

\bibitem{shoval_alon_sontag_2011}
O.~Shoval, U.~Alon, and E.D. Sontag.
\newblock Symmetry invariance for adapting biological systems.
\newblock {\em SIAM Journal on Applied Dynamical Systems}, 10:857--886, 2011.

\bibitem{shoval10}
O.~Shoval, L.~Goentoro, Y.~Hart, A.~Mayo, E.D. Sontag, and U.~Alon.
\newblock Fold change detection and scalar symmetry of sensory input fields.
\newblock {\em Proc.\ Natl.\ Acad.\ Sci.\ U.S.A.}, 107:15995--16000, 2010.

\bibitem{TuShimizuBerg2008}
Y.~Tu, T.~S. Shimizu, and H.~C. Berg.
\newblock {{M}odeling the chemotactic response of {E}scherichia coli to
  time-varying stimuli}.
\newblock {\em Proc.\ Natl.\ Acad. Sci.\ U.S.A.}, 105:14855--14860, 2008.

\end{thebibliography}

\end{document}